\documentclass[11pt,reqno,a4paper]{amsart}

\usepackage[toc,page]{appendix}
\usepackage[utf8]{inputenc}
\usepackage{csquotes}
\usepackage{MnSymbol}
\usepackage{esint}
\usepackage{enumitem}
\usepackage[english]{babel}
\usepackage{amsmath}
\usepackage{thm-restate}
\usepackage{graphicx}
\usepackage{mathtools}
\usepackage{xcolor}

\usepackage{biblatex} 
\usepackage[
	colorlinks,
	pdfpagelabels,
	pdfstartview = FitH,
	bookmarksopen = true,
	bookmarksnumbered = true,
	linkcolor = blue,
	plainpages = false,
	hypertexnames = false,
	citecolor = red] {hyperref}

\hypersetup{
	linktoc=page
}

\usepackage[capitalize]{cleveref}
\crefname{enumi}{}{}
\crefname{equation}{}{}

\makeatletter
\def\@tocline#1#2#3#4#5#6#7{\relax
  \ifnum #1>\c@tocdepth 
  \else
    \par \addpenalty\@secpenalty\addvspace{#2}%
    \begingroup \hyphenpenalty\@M
    \@ifempty{#4}{%
      \@tempdima\csname r@tocindent\number#1\endcsname\relax
    }{%
      \@tempdima#4\relax
    }%
    \parindent\z@ \leftskip#3\relax \advance\leftskip\@tempdima\relax
    \rightskip\@pnumwidth plus4em \parfillskip-\@pnumwidth
    #5\leavevmode\hskip-\@tempdima
      \ifcase #1
       \or\or \hskip 1em \or \hskip 2em \else \hskip 3em \fi%
      #6\nobreak\relax
    \dotfill\hbox to\@pnumwidth{\@tocpagenum{#7}}\par
    \nobreak
    \endgroup
  \fi}
\makeatother

\newtheorem{theorem}{Theorem}

\newtheorem*{theorem*}{Theorem}
\newtheorem{proposition}{Proposition}[section]
\newtheorem{lemma}[proposition]{Lemma}

\theoremstyle{definition}
\newtheorem{definition}[proposition]{Definition}

\newtheorem{example}[proposition]{Example}
\numberwithin{equation}{section}

\def \R {\mathbb {R}}
\def \N {\mathbb {N}}
\def \E {\mathbb{E}}
\def \P {\mathbb{P}}
\def \V {\mathbb{V}}

\renewcommand{\hat}{\widehat}
\renewcommand{\tilde}{\widetilde}

\newcommand{\eps}{\varepsilon}

\allowdisplaybreaks

\begin{document}

\title[Sequential inductive prediction intervals]{Sequential inductive prediction intervals}
\author[Avelin]{Benny Avelin}
\address{Combient Mix (Silo AI) \\ Mäster Samuelsgatan 56, 111 21 Stockholm, Sweden}
\email{benny.avelin@silo.ai}


\date{\today}
\begin{abstract}
	In this paper we explore the concept of sequential inductive prediction intervals using theory from sequential testing. We furthermore introduce a 3-parameter PAC definition of prediction intervals that allows us via simulation to achieve almost sharp bounds with high probability.
\end{abstract}

\maketitle

\section{Coverage and prediction intervals}
If we let $X$ correspond to a prediction, then a $1-\alpha$ prediction interval is in the simplest case nothing other than a set $C$ such that the measure of that set is exactly $1-\alpha$, i.e.
\begin{align*}
	\P(X \in C) = 1-\alpha.
\end{align*}
The set $C$ is unknown, and we would like to construct it based on observations. Lets us consider a sequence of i.i.d.~random variables $X_1,\ldots, X_n$ then say we have a way to determine $C ( X_1,\ldots,X_{n-1})$ as a function of the first $n-1$ observations, then the requirement of being a prediction interval is usually
\begin{align}\label{eq:predguarantee1}
	\P(X_n \in C(X_1,\ldots,X_{n-1})) = 1-\alpha.
\end{align}
In this simple setting we could guarantee exactly the $1-\alpha$ level if we take $C = [\hat q_{\alpha/2}, \hat q_{1-\alpha/2}]$, where $\hat q_{\delta}$ is the empirical $\delta$ quantile, with this choice and if $(n-1) \alpha/2 \in \N$ then \cref{eq:predguarantee1} holds. This is called a marginal coverage guarantee.

In the context of regression the prediction interval statement becomes for a sequence of i.i.d. observations $(X_1,Y_1),\ldots,(X_n,Y_n)$
\begin{align*}
	\P(Y_n \in C((X_1,Y_1),\ldots,(X_{n-1},Y_{n-1}),X_n)) = 1-\alpha.
\end{align*}
To achieve this, assume we have been given some data and say that we have found a function $f(X)$ approximating the regression function $\E[Y \mid X]$, then if we define the residual $Z_i = Y-f(X_i)$ we can form the prediction interval as $f(X_n) + C$ where $C$ is the same quantile based interval as in the previous example, but this time based on $Z_1,\ldots,Z_{n-1}$.

In the framework of conformal prediction \cite{VGS, PPVG, LW}, one would construct a well-designed non-conformity score. The non-conformity score is in the formal mathematical sense a function from the bag of previous observations and the new observation to a real number. A simple example would be to take what empirical quantile the new observation belongs to if we add it to the bag of previous observations.
 
An extension of the idea of marginal guarantee is that of a sequential guarantee. Consider again a sequence of random variables $Z_1,\ldots,Z_n$ and let's say we have a sequence of prediction intervals $C_i = C(Z_1,\ldots,Z_i)$, the transductive prediction problem (\cite{VovkOriginal}) is to guarantee that
\begin{align*}
	\P(Z_{i} \in C_{i-1}) = \alpha
\end{align*}
for all $i$ and that
the sequence of errors
\begin{align*}
	E_i = \mathbb{I}_{Z_i \in C_{i-1}}
\end{align*}
is an i.i.d. sequence of random variables. With a simple randomization trick (to avoid the discreteness of the problem) we can essentially use the empirical quantiles here, see~\cite{Vovk12}.

The interesting thing with respect to this paper, is that we can make this guarantee, but only for the next sample. If we instead would like to guarantee that
\begin{align}\label{eq:inductive}
	\P(\P(Z_i \in C(Z_1,\ldots,Z_{i-1}) \mid Z_1,\ldots,Z_{i-1}) \geq 1-\alpha) \geq 1-\delta
\end{align}
this is denoted as the conditional coverage and is phrased as a 2-parameter PAC (Probably Approximately Correct,~\cite{PAC}) problem, essentially introduced in~\cite{Vovk12}. What is interesting here, is that if we fix the value of $i$ then this denotes that with probability greater than $1-\delta$ we can, given the calibration data $Z_1,\ldots,Z_{i-1}$, construct an interval that has at least the correct coverage probability. We interpret this to mean that the prediction interval holds for an arbitrary number of future observations, which is the so-called inductive setting (see~\cite{Vovk12}). That is, we want to construct an interval which with high probability can be trusted.

Let us consider a scenario that highlights the differences between the transductive and the inductive form of prediction intervals. In the first we create a sequence of intervals for which the next sample will land with a guaranteed probability, which is only true for one observation. This would be useful in the online setting when you are trying to bet on the next upcoming value in the context of say finance.
This might seem subtle, but the fact that the consecutive errors $E_i$ are i.i.d.~does not guarantee that the probability of the next sample being in the interval given the previous observations is anything exact or above $1-\alpha$, if we have many samples it will however be close with high probability. 

In the second setting, i.e.~when we talk about the inductive form, we are trying to create an interval that we can trust for all eternity, or at least for a fixed number of new samples which was considered in \cite{Vovk12}. This is different because we have instead phrased it as a decision problem. Namely, if the procedure we design produces a prediction interval, then if we say that this interval has a coverage of at least $1-\alpha$ we would be right $1-\delta$ of the time. This is very similar to producing a confidence interval, where the procedure of computing the interval and deciding that the true parameter lies in the computed interval, we would be right $1-\delta$ of the time.

In terms of selling a product where we either train the model once and sell it or are intermittently updating the model with new data, we cannot control how many times the customer will use the model for predictions. This means that the transductive setting is not useful and this second style of producing inductive intervals makes more sense.

The statement in \cref{eq:inductive} considers that we have chosen a sample size $i$ and construct our prediction interval based on that sample size. This does not leverage the fact that as you see more testing points you could tighten the prediction interval to be closer to the true quantiles of the distribution. This is intimately close to the idea of sequential testing, and by the duality of tests and confidence intervals, the construction of confidence sequences for parameters.

A confidence sequence is something that satisfies
\begin{align*}
	\P(\forall n: \text{parameter in set defined by } (X_1,\ldots,X_n)) \geq 1-\delta.
\end{align*}
For confidence intervals this would be a sequence of confidence intervals constructed such that the probability that our parameter is in the set is at least $1-\delta$ for all $n$ uniformly. This means that we are allowed to peek at our test at all times, the price we pay for this peeking is that the intervals are necessarily bigger.

The same idea can be phrased for prediction, in this case we call $C_n$ the prediction intervals and $X$ is then the true value
\begin{align*}
	\P \left (\forall n: \P(X \in C_n) \geq 1-\alpha \right ) \geq 1-\delta.
\end{align*}
That is, our sequence of prediction intervals will have coverage at least $1-\alpha$ with probability greater than $1-\delta$. This statement would be moot if we did not require $C_n$ to get tighter as $n \to \infty$, specifically we would like $\P(X \in C_n) \to 1-\alpha$ as $n \to \infty$ with the best possible rate.

\subsection{Confidence sequence for the quantile and prediction intervals}

We will base our discussion on a recent result \cite{HowRam}.

Let $X$ be a real valued random variable. Write $F(x) := \P(X \leq x)$ as the distribution function, define the empirical distribution of an i.i.d.~sequence $X_1,\ldots,X_n$ sampled from $F$ as $\hat F_n(x) = \frac{1}{n}\sum_{i=1}^n \mathbb{I}_{X_i \leq x}$. Define the upper and lower quantile functions for $p \in [0,1]$,
\begin{align*}
	Q(p) &:= \sup\{x: F(x) \leq p\} \\
	Q^-(p) &:= \sup\{x: F(x) < p\},
\end{align*}
respectively. Denote $\hat Q_n(p)$ and $\hat Q_n^-(p)$ the upper and lower quantile functions w.r.t. $\hat F_n$.

Let $Q$ be the upper quantile function, let $\hat F$ be the empirical distribution function and $\hat Q$ be the empirical quantile function. Superscript $-$ denotes the corresponding lower quantile functions.

\begin{theorem}[\cite{HowRam}] \label{thm:1}
	Let $Z_1,\ldots \overset{IID}\sim F$, and let $Q$ be the corresponding quantile function. Then, for a fixed $\delta \in (0,1)$ the following holds
	\begin{align*}
		\P(\forall i \in \N, \forall p \in (0,1): Q(p) \in [\hat Q_i(p-g_i),\hat Q_i^-(p+g_i)]) \geq 1-\delta,
	\end{align*}
	where
	\begin{align*}
		g_i = 0.85 \sqrt{\frac{1}{i} [\log(\log(ei)) + 0.8 \log(1612/\delta)]}
	\end{align*}

\end{theorem}

The good thing about the above is that the probability is over all $i$ and over all $p$-quantiles simultaneously (this is because it is a consequence of a DKW style inequality, i.e.~a bound on the $L^\infty$ norm of $F-\hat F$). We can leverage this fact to get stronger control and obtain also an upper bound on the coverage in \cref{eq:inductive}.
We will start with the lower bound as it is easier. Choose a confidence level $\alpha$ and a failure rate $\delta$ and note that \cref{thm:1} implies
\begin{align}\label{eq:quantile0}
	\P\left(\forall i, \alpha \in (0,1): \hat Q_i\left(\frac{\alpha}{2}-g_i\right)\leq Q\left(\frac{\alpha}{2}\right); Q\left(1-\frac{\alpha}{2}\right) \leq \hat Q_i^-\left(1-\frac{\alpha}{2}+g_i\right) \right) \geq 1-\delta,
\end{align}
%
which implies that if we define the two intervals
\begin{align*}
	C^{l,\alpha}_i &= [\hat Q_i(\alpha/2-g_i),\hat Q_i^-(1-\alpha/2+g_i)]\\
	C^{u,\alpha}_i &= [\hat Q_i^-(\alpha/2+g_i),\hat Q_i(1-\alpha/2-g_i)]
\end{align*}
then
\begin{align} \label{eq:prediction_interval_lower}
	\P(\forall i, \alpha \in (0,1): \P(Z \in C^{l,\alpha}_i) \geq 1-\alpha) \geq 1-\delta,
\end{align}
the upper bound can be obtained similarly, saying that
\begin{align*}
	\P(\forall i,\alpha \in (0,1): \P(Z \in C^{u,\alpha}_i) \leq 1-\alpha) \geq 1-\delta.
\end{align*}
Note that from \cref{eq:quantile0} tells us that the upper and lower bound holds simultaneously with probability greater than $1-\delta$, so we don't have to use the union bound to combine them.
This provides good control over the coverage probability, the issue is that we obtain two sets, one bigger with the lower bound and one smaller with the upper bound. We can rephrase this for a single interval, where we get an upper and lower bound, however we have to choose if we want a deterministic lower bound and a random upper bound, or we can choose a deterministic upper bound and a random lower bound. 

Let us show how we do this, we begin by defining
\begin{align*}
	\alpha^l_i := \sup \{\tilde \alpha \in [0,1]: \hat Q^-_i(\tilde \alpha/2+g_i) < \hat Q_i(\alpha/2-g_i)\}
\end{align*}
where the supremum is $0$ if the set is empty. We also define
\begin{align*}
	\alpha^u_i := \sup \{\tilde \alpha \in [0,1]: \hat Q_i(1-\tilde \alpha/2-g_i) > \hat Q_i(1-\alpha/2+g_i)\}
\end{align*}
where the supremum is $0$ if the set is empty. Then define 
\begin{align*}
	\alpha^\ast_i(\alpha) = \min(\alpha^l_i(\alpha),\alpha^u_i(\alpha))	.
\end{align*}
From this construction we get that $C_i^{l,\alpha} \subset C_i^{u,\alpha^\ast_i(\alpha)}$, which together with the fact that \cref{eq:quantile0} holds for all $\alpha$, we get the following theorem

\begin{theorem} \label{thm:prediction}
	We have for the observable value $\alpha^\ast_i$ that
	\begin{align*}
		\P(\forall i,\alpha \in (0,1): 1-\alpha^\ast_i(\alpha) \geq \P(Z \in C^{l,\alpha}_i) \geq 1-\alpha) \geq 1-\delta.
	\end{align*}
\end{theorem}

This is a very useful estimate, as we can guarantee the lower bound with probability $1-\delta$, the upper bound we can determine using the data. 

It is also clear that studying the definition, we see that $\alpha_i^\ast \to \alpha$ as $i \to \infty$ since the empirical quantiles converge a.s.~to the true quantile and $g_i \to 0$.

\subsection{Conformal prediction and confidence sequences} \label{sec:conformal}
Let us now show how to use confidence sequences in the context of an inductive prediction setting using ideas from conformal prediction. Namely, assume that we have a score function $S(x,y)$ which measures the error we make. Typical choice is
\begin{align*}
	S(x,y) = \frac{|f(x)-y|}{\sigma(x)}
\end{align*}
where $f$ is the function making the prediction and $\sigma$ could for instance be the square root of an estimator of the conditional variance. For a specific example see the kernel variance estimator in \cref{appendix:d.nadaraya.watson.variance}.

Now, for an i.i.d~sequence of calibration points we will create the random variable $Z_i = S(X_i,Y_i)$ and for this we can create the confidence sequence
\begin{align}\label{eq:conformal1}
	\P(\forall i: \P(Z \in C_i(\alpha) \mid Z_1,\ldots,Z_{i}) \geq 1-\alpha) \geq 1-\delta.
\end{align}
This confidence sequence $C_i(\alpha)$ can be created using for instance \cref{thm:prediction}.
Now note that $C_i(\alpha) = [L_i(\alpha),U_i(\alpha)]$ is an interval, therefore if we define
the score function to be the signed version of $S$ above, i.e.
\begin{align*}
	S(x,y) = \frac{y-f(x)}{\sigma(x)}
\end{align*}
then we can define the scaled and shifted interval
\begin{align*}
	D_i(\alpha) = [f(X) + \sigma(X) L_i(\alpha), f(X) + \sigma(X) U_i(\alpha)],
\end{align*}
and \cref{eq:conformal1} becomes
\begin{align*}
	\P(\forall i: \P(Y \in D_i(\alpha) \mid (X_1,Y_1),\ldots,(X_i,Y_i)) \geq 1-\alpha) \geq 1-\delta
\end{align*}
The fact that our score function is signed allows us to treat problems in where the conditional distribution of $Y \mid X$ is skew, see \cref{fig:prediction_interval1D_skew}.

\subsection{IID bounded residual}

Let us see another application of these intervals in the context where we cannot guarantee identical distribution of the residual. Assume that we have a heteroscedastic problem, i.e.
\begin{align*}
	Y = f(X) + Z
\end{align*}
and assume that $Z$ is mean zero and independent of $X$. Then $f$ is the regression function. Assume that we have used data to estimate $f$, call the estimate $\hat f$, and assume that we managed to achieve $\|f - \hat f\|_{L^\infty} < \eps$, which can be achieved with high probability if for instance $f$ is Lipschitz continuous. We are interested in this situation to create prediction intervals around $Y$ using $\hat f$, but where we take deterministic $x$. 

Consider now a sequence of given points $x_1,\ldots,x_n$, then $Y_i = f(x_i) + Z_i$. The residual on the other hand is $R_i = Y_i - \hat f(x_i)$, and since $\hat f \neq f$ this is only a sequence of independent random variables, it is not identically distributed. We wish to create a prediction interval that contains $R_i$ with high probability. Therefore, we introduce the following concept
\begin{definition}
	Let $X_1,\ldots$ be an independent sequence of random variables and assume that there exists constants $c_1,\delta_1,c_2,\delta_2$ and an IID sequence of random variables $Y_1,\ldots$ (a coupling) such that
	\begin{align*}
		c_1 Y_i + \delta_1 \leq X_i \leq c_2 Y_i + \delta_2.
	\end{align*}
	Then we say that $X_1,\ldots$ is an \emph{IID bounded sequence}.
\end{definition}

The above definition allows us to compare the distribution functions
\begin{lemma} \label{lem:comparison:essIID}
	Let $X_1,\ldots$ be an IID bounded sequence with constants $c_1$, $\delta_1$, $c_2$, $\delta_2$ and IID bounding sequence $Y_1,\ldots$. Then
	\begin{align*}
		F_{X_i} (c_1 x + \delta_1) \leq F_{Y_i} (x) \leq F_{X_i} (c_2 x + \delta_2)
	\end{align*}
	Define the empirical distribution functions as $\hat F_{Y_n}(x) = \frac{1}{n} \sum_{i=1}^n \mathbb{I}_{Y_i \leq x}$ and
	$\hat F_{X_n}(x) = \frac{1}{n} \sum_{i=1}^n \mathbb{I}_{X_i \leq x}$, then
	\begin{align*}
		\hat F_{X_n} (c_1 x + \delta_1) \leq \hat F_{Y_n} (x) \leq \hat F_{X_n} (c_2 x + \delta_2)
	\end{align*}
\end{lemma}
\begin{proof}
	It is immediate from the definition of IID bounded sequences, that
	\begin{align*}
		F_{Y_i} (x) = \P(Y_i \leq x) \leq \P((X_i - \delta_2)/c_2 \leq x) =
		F_{X_i} (c_2 x + \delta_2).
	\end{align*}
	The result for the empirical distribution function follows in the same way.
\end{proof}

\begin{lemma} \label{lem:quantile:essentialIID}
	Let the sequence $X_i$ be as in \cref{lem:comparison:essIID}. Then define the quantile function $Q_{X_i}(p) = \sup\{x : F_{X_i}(x) \leq p\}$ and the lower quantile function $Q^{-}_{X_i}(p) = \sup\{x : F_{X_i}(x) < p\}$. We define $\hat Q_{X_n}$ and $\hat Q^{-}_{X_n}$ similarly.
	Then, for $\alpha \in (0,1)$ with $n\alpha \in \N$, it holds
	\begin{align*}
		c_1 \hat Q_{Y_n}(\alpha) + \delta_1 \leq \hat Q_{X_n}(\alpha) \leq c_2 \hat Q_{Y_n}(\alpha)+\delta_2,
	\end{align*}
	and
	\begin{align*}
		c_1 \hat Q^-_{Y_n}(\alpha) + \delta_1 \leq \hat Q^-_{X_n}(\alpha) \leq c_2 \hat Q^-_{Y_n}(\alpha)+\delta_2.
	\end{align*}
\end{lemma}
\begin{proof}
	Let $n\alpha \in \N$ and $\alpha \in (0,1)$, then by \cref{lem:comparison:essIID}
	\begin{align*}
		\alpha+1/n = \hat F_{Y_n} (\hat Q_{Y_n}(\alpha)) \leq \hat F_{X_n}(c_2 \hat Q_{Y_n}(\alpha)+\delta_2)
	\end{align*}
	this implies that $c_2 \hat Q_{Y_n}(\alpha)+\delta_2 \not \in \{x : \hat F_{X_n}(x) \leq \alpha\}$, and as such we obtain
	\begin{align*}
		\hat Q_{X_n}(\alpha) \leq c_2 \hat Q_{Y_n}(\alpha)+\delta_2.
	\end{align*}
	On the other hand we have
	\begin{align*}
		\alpha+1/n = \hat F_{X_n} (\hat Q_{X_n}(\alpha)) \leq \hat F_{Y_n}\left (\frac{\hat Q_{X_n}(\alpha)-\delta_1 }{c_1}\right)
	\end{align*}
	which in the same way implies that
	\begin{align*}
		c_1 \hat Q_{Y_n}(\alpha) + \delta_1 \leq \hat Q_{X_n}(\alpha).
	\end{align*}
	Let us now consider the lower quantile function, we first note that
	\begin{align*}
		\alpha = \hat F_{Y_n} (\hat Q^-_{Y_n}(\alpha)) \leq \hat F_{X_n}(c_2 \hat Q^-_{Y_n}(\alpha)+\delta_2)
	\end{align*}
	which again implies that $c_2 \hat Q^-_{Y_n}(\alpha)+\delta_2 \not \in \{x : \hat F_{X_n}(x) < \alpha\}$ hence
	\begin{align*}
		\hat Q^-_{X_n}(\alpha) \leq c_2 \hat Q^-_{Y_n}(\alpha)+\delta_2,
	\end{align*}
	the lower bound now follows in the same way as above.
\end{proof}

Now, for IID bounded sequences we can construct prediction intervals similarly to IID sequences but at a cost.

\begin{theorem} \label{thm:essentialIID}
	Define the interval
	\begin{align*}
		C_n = \left [\frac{\hat Q_{X_n}(\frac{\alpha}{2}-g_i)-\delta_2}{c_2},\frac{\hat Q_{Y_n}^-(1-\frac{\alpha}{2}+g_i)-\delta_1}{c_1} \right ]
	\end{align*}
	then
	\begin{align*}
		\P(\P(Y_{n+1} \in C_n \mid X_1,\ldots,X_n) \geq 1-\alpha) \geq 1-\delta.
	\end{align*}
	Correspondingly, by enlarging $C_n = (l,r)$ as follows
	\begin{align*}
		\tilde C_n = \left ( c_1 l + \delta_1, c_2 r + \delta_2 \right )
	\end{align*}
	we obtain a similar statement for $X_{n+1}$
	\begin{align*}
		\P(\P(X_{n+1} \in \tilde C_n \mid X_1,\ldots,X_n) \geq 1-\alpha) \geq 1-\delta.
	\end{align*}
\end{theorem}
\begin{proof}
	By \cref{thm:prediction} we have the prediction interval
	\begin{align*}
		[\hat Q_{Y_n}(\alpha/2-g_i),\hat Q_{Y_n}^-(1-\alpha/2+g_i)].
	\end{align*}
	This interval will with probability $\delta$ have coverage greater than $1-\alpha$. Now from \cref{lem:quantile:essentialIID} we see that we can cover the prediction interval above as follows
	\begin{align*}
		\left [\hat Q_{Y_n}(\frac{\alpha}{2}-g_i),\hat Q_{Y_n}^-(1-\frac{\alpha}{2}+g_i) \right ] \subset \left [\frac{\hat Q_{X_n}(\frac{\alpha}{2}-g_i)-\delta_2}{c_2},\frac{\hat Q_{Y_n}^-(1-\frac{\alpha}{2}+g_i)-\delta_1}{c_1} \right ].
	\end{align*}
	To get the last statement, we simply apply \cref{lem:quantile:essentialIID} again.
\end{proof}

\begin{example}
	For the day index $t$ we sell $X_t$ number of items, and $\E[X_t] = \lambda(t)$. We assume that we have used previous years data to estimate the rate function $\lambda(t)$ which we can with high probability do with uniform error $\eps$.

	Our observations are thus
	\begin{align*}
		X_t = \lambda(t)+Z_t
	\end{align*}
	where $Z_t$ are i.i.d.

	Our estimated rate function is $\hat \lambda(t)$ and we know that $|\lambda(t) - \hat \lambda(t)| \leq \eps$. Let us construct the sequence of errors defined as
	\begin{align*}
		E_t = X_t - \hat \lambda(t) = Z_t + \lambda(t)-\hat \lambda(t) =: Z_t + \eps_t,
	\end{align*}
	with $|\eps_t| \leq \eps$.

	We are now exactly in the setting where we can apply \cref{thm:essentialIID}.
\end{example}

\section{Sharper bounds}

In this section we will use simulation in order to sharpen the bound in \cref{thm:1}. This is especially the case for smaller values of $t$, as can be seen in \cref{fig:motivation}, where we have the correct shape but slightly inflated values. Instead of trying to come up with a smarter way to estimate the probabilities involved we can simulate (at least for finite values).

\begin{figure}
	\begin{center}
		\includegraphics[scale=0.5]{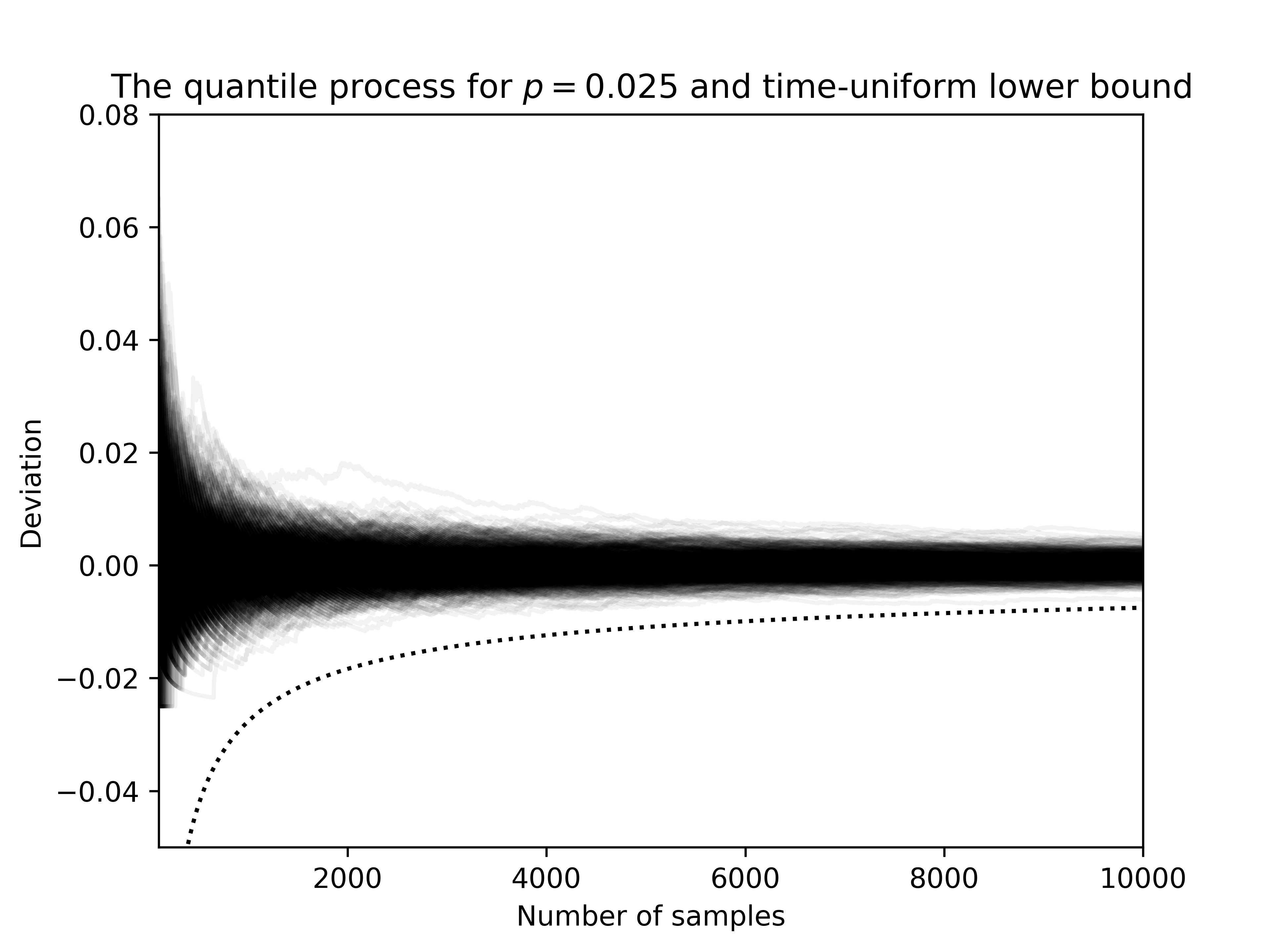}
	\end{center}
	\caption{Some trajectories of the quantile process and the time-uniform bound provided in \cite{HowRam}.}
	\label{fig:motivation}
\end{figure}

To begin, we recall the notation of \cite{HowRam} and define
\begin{align}\label{eq:empirical_process}
	S_t = \sum_{i=1}^t (\mathbb{I}_{X_i \leq Q(p)}-p)
\end{align}
where $X_i \sim F$ is an i.i.d.~sequence and $Q(p)$ is the $p$ quantile w.r.t.~$F$. We note that by definition $\hat F_t(Q(p)) = S_t/t$, as such, if we can find a function $-p < f(t,p,\delta) < 0$ such that
\begin{align} \label{eq:empirical_process_bound}
	\P(\exists t:\, S_t/t < f(t,p,\delta)) \leq \delta
\end{align}
it would imply by the definition of the quantile function, that
\begin{align}\label{eq:empirical_process_implies_quantile}
	\P(\exists t:\, Q(p) < \hat Q_t(p+f(t,p,\delta))) \leq \delta.
\end{align}
Thus, if we wish to have a lower bound on a quantile (think small quantile), we can study the Bernoulli process $S_t$. In order to explain the simulation setup we begin by requiring that we have an initial waiting time $m$ and a cutoff-time $M$ (for the simulation), and consider
\begin{multline*}
	\P(\exists t:\, S_t/t < f(t,p,\delta)) \leq \P(\exists t \in [m,M]:\, S_t/t < f(t,p,\delta))\\
	 + \P(\exists t > M:\, S_t/t < f(t,p,\delta))
\end{multline*}
The first term we will handle via simulation and the second term will be analytically bounded using the following quantile specific version of \cref{thm:1}, which can be found in~\cite{HowRam}:
\begin{theorem} \label{thm:HowRamSingleQuantileDecay}
	For any $p,\delta \in (0,1)$, $s,\eta > 1$, and $M \geq 1$, we have
	\begin{align*}
		\P(\exists t \geq M: Q(p) < \hat Q_t(p-g(t,M,p,\delta))) \leq \delta
	\end{align*}
	with
	\begin{align*}
		g(t,M,p,\delta) = \sqrt{k_1^2 p (1-p) t l(t) + k_2^2 c_p^2 l^2(t)} + c_p k_2 l(t)
	\end{align*}
	where
	\begin{align*}
		\begin{cases}
			l(t) = s \log(\log(\eta t/M)) + \log(2 \zeta(s)/(\delta \log^s \eta)) \\
			k_1 = (\eta^{1/4} + \eta^{-1/4})/\sqrt{2} \\
			k_2 = (\sqrt{\eta} + 1)/2 \\
			c_p = (1-2p)/3,
		\end{cases}
	\end{align*}
	where $\zeta$ is the Riemann zeta function.
\end{theorem}

\begin{figure}
	\begin{center}
		\includegraphics[scale=0.5]{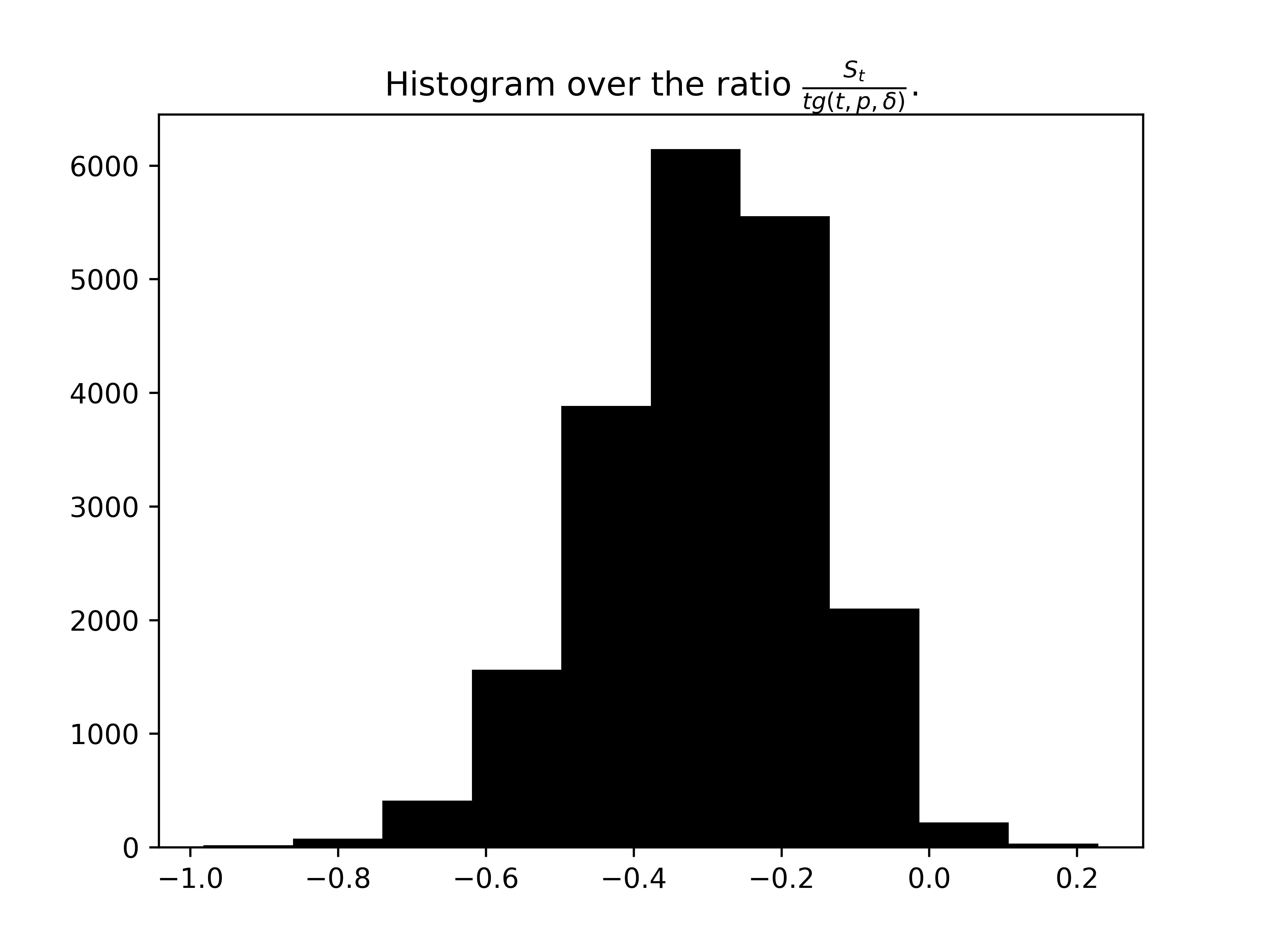}
	\end{center}
	\caption{Histogram over $c_Z$.} \label{fig:histogram_quantile_process}
\end{figure}

We now start by defining the random variable
\begin{align*}
	Z = \min_{t \in [m,M]} \frac{S_t}{t g(t,p,\delta/2)}
\end{align*}
where the function $g$ is given in \cref{thm:HowRamSingleQuantileDecay} with some fixed parameters, we have chosen the same as in \cite{HowRam}, namely $\eta = 2.04$ and $s=1.4$. The goal of the simulation is to estimate the $\delta/2$ quantile of $Z$ (see \cref{fig:histogram_quantile_process}). If we consider an empirical quantile $\hat Q_{t,Z} (\delta/2)$, formed using $t$ i.i.d.~copies of $Z$ which we denote $\mathbf{Z}_t$, then using \cref{lem:quantile_confidence} we get that for any $\tilde \delta \in (0,1)$,
\begin{align*}
	\P(\P(\tilde Z \leq \hat Q_{t,\mathbf{Z}_t}(\delta/2-\epsilon(t,\delta/2,\tilde \delta)) \mid \mathbf{Z}_t) \leq \delta/2) \geq 1-\tilde \delta
\end{align*}
where $\tilde Z$ is an independent copy of $Z$. Translating this back to the quantile process we obtain
\begin{align*}
	\P\bigg [\P\big (\exists t \in [m,M]:\, \hat F_t(Q(p)) < \hat Q_{t,\mathbf{Z}_t}\bigg (\frac{\delta}{2}-\epsilon\big (t,\frac{\delta}{2},\tilde \delta\big )\bigg ) g\big (t,p,\frac{\delta}{2}\big ) \mid \mathbf{Z}_t \big ) \leq \frac{\delta}{2} \bigg ] \geq 1-\tilde \delta.
\end{align*}
Thus, the introduction of the simulation step makes the quantile confidence statement of a $2$-parameter PAC type, as described for the prediction interval. If we later employ this in the same way as we did before but for the prediction interval we will have a $3$-parameter PAC statement. We can see how this sharpens the bounds in \cref{fig:simulation_quantile}.

\begin{figure}
	\begin{center}
		\includegraphics[scale=0.5]{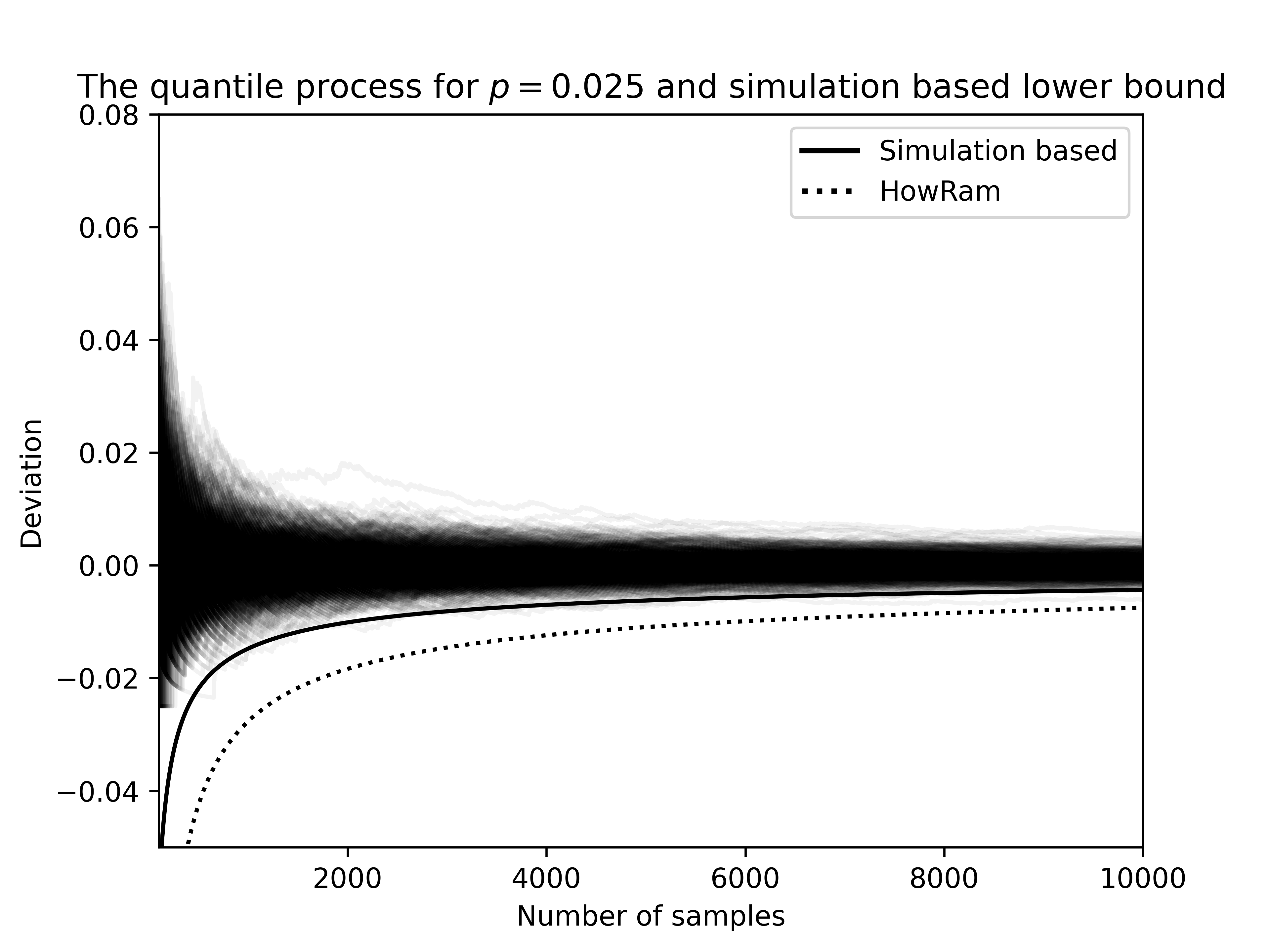}
	\end{center}
	\caption{Quantile process with a simulated quantile compared to the bound obtained by \cref{thm:HowRamSingleQuantileDecay}.} \label{fig:simulation_quantile}
\end{figure}

If we denote $c_Z = -\hat Q_{t,\mathbf{Z}_t}(\delta/2 - \epsilon(t,\delta/2,\tilde \delta))$ (think of $c_Z > 0$) and $g_t = g(t,p,\delta/2)$, then the above can be written more compactly as
\begin{align*}
	\P(\P(\exists t \in [m,M]:\, S_t/t < -c_Z g_t \mid \mathbf{Z}_t) \leq \delta/2) \geq 1-\tilde \delta.
\end{align*}
Define now the $\mathbf{Z}_t$ dependent event
\begin{align*}
	A = \{Q_Z(\delta/2) \geq \hat Q_{t,\mathbf{Z}_t}(\delta/2 - \epsilon(t,\delta/2,\tilde \delta))\},
\end{align*}
then, on $A$ we have
\begin{align*}
	\P(\exists t \in [m,M]:\, S_t/t < -c_Z g_t) \leq \delta/2.
\end{align*}
Using the union bound we can estimate on the event $A$
\begin{multline*}
	\P(\exists t \geq m: S_t/t < -\mathbb{I}_{t \leq M} c_Z g_t - \mathbb{I}_{t > M} g(t,M,p,\delta/2)) \\
	\leq \P(\exists t \in [m,M]: S_t/t < -c_Z g_t) + \P(\exists t > M: S_t/t < -g(t,M,p,\delta/2)) \leq \delta
\end{multline*}

\begin{example}
	For $p=0.05, \delta = 0.05,m=10, \tilde \delta = 0.0000001$, and $M=10000$ we obtain that $c_Z$ is roughly $0.6$ and that parameters $\eta,s$ can be chosen such that $g(M,M,p,\delta/2)$ is essentially the same size as $c_Z g_M$. 
\end{example}

Assembling everything we finally get,
\begin{align} \label{eq:2PAC_quantile}
	\P(\P(\exists t \geq m: \hat F_t(Q(p)) < p-\mathbb{I}_{t \leq M} c_Z g_t - \mathbb{I}_{t > M} g(t,M,p,\delta/2)\mid \mathbf{Z}_t) \leq \delta) \geq 1-\tilde \delta
\end{align} 

We can compare the bounds obtained in \cref{fig:simulation_stitch} for our particular choices.

\begin{figure}
	\begin{center}
		\includegraphics[scale=0.5]{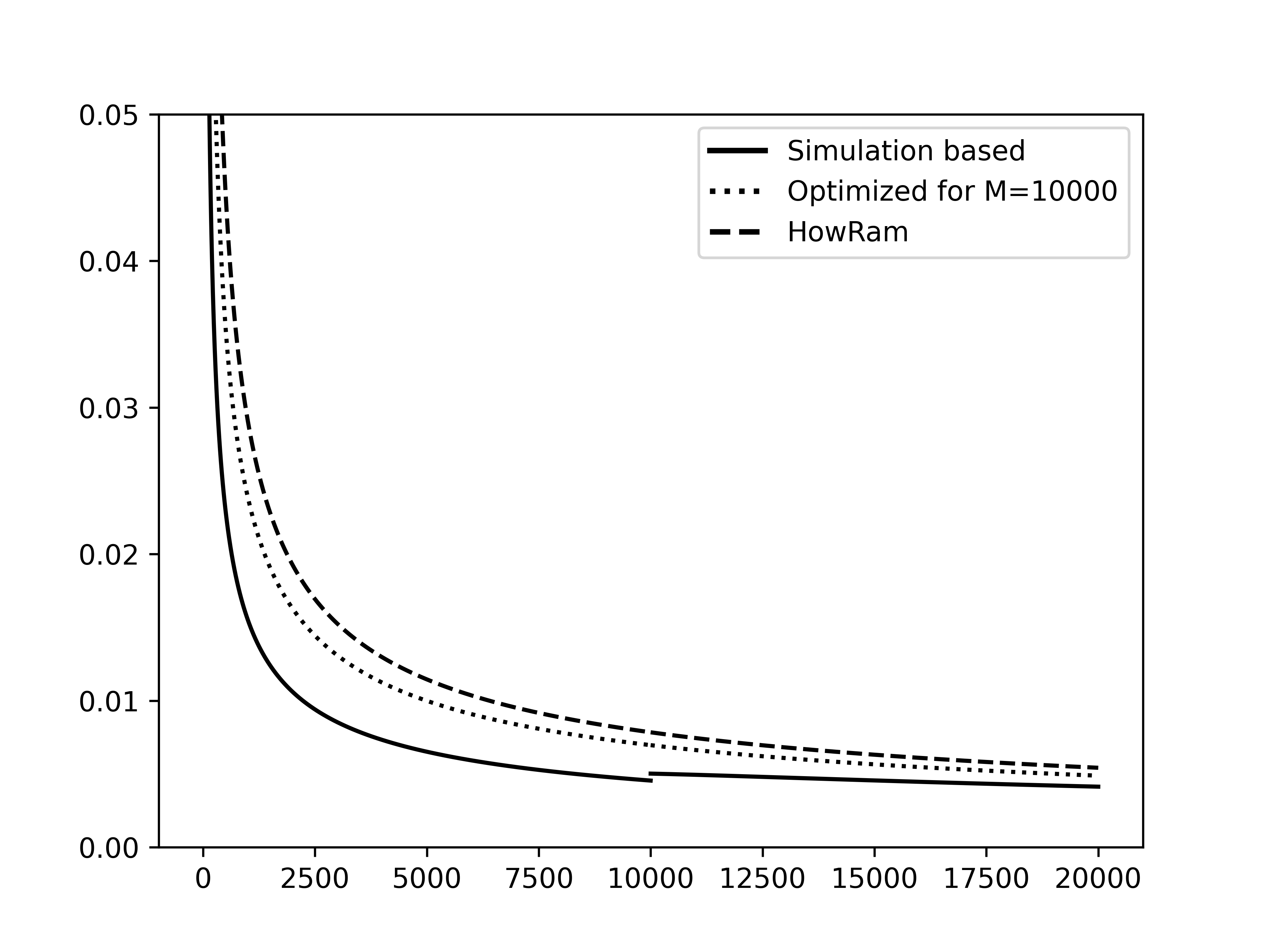}
	\end{center}
	\caption{Comparison of the different bounds.} \label{fig:simulation_stitch}
\end{figure}

\subsection{$3$-parameter PAC for prediction intervals} \label{sec:3PAC}
Combining the arguments leading to \cref{eq:2PAC_quantile} and \cref{eq:prediction_interval_lower} we can redefine the event $A$ to be (where we have defined $Z$ w.r.t. $p = \alpha/2$)
\begin{align*}
	A = \{Q_Z(\delta/4) \geq \hat Q_{t,\mathbf{Z}_t}(\delta/4 - g(t,\delta/4,\tilde \delta))\},
\end{align*}
then the event $A$ happens with probability $1-\tilde \delta$. On the event $A$ we then obtain that our prediction interval can be taken to be
\begin{align} \label{eq:3PAC_prediction}
	C^{l,\alpha}_i &= [\hat Q_i(\alpha/2-|c_Z| g_i),\hat Q_i^-(1-\alpha/2+|c_Z| g_i)]
\end{align} 
and on the event $A$ we have the following version of \cref{eq:prediction_interval_lower}
\begin{align} \label{eq:prediction_interval_lower_3PAC}
	\P(\forall i: \P(X \in C^{l,\alpha}_i) \geq 1-\alpha) \geq 1-\delta.
\end{align}

This constitutes a $3$-parameter PAC definition with parameters $(\tilde \delta,\delta,\alpha)$, where the first can be taken as small as we wish at the cost of our simulation.

\section{Application examples} 

In this section we will apply our $3$-parameter PAC formulation for prediction intervals in conjunction with what is developed in \cref{sec:conformal}. The first example is the following regression problem
\begin{align*}
	r(x) := \E[Y \mid X=x] = \sin(3x)
\end{align*}
\begin{align*}
	\sigma(x) := \sqrt{\mathbb{V}[Y \mid X=x]} = |\sin(10x)|.
\end{align*}
We assume that
\begin{align*}
	Y = r(X) + \sigma(X) \epsilon, \quad \epsilon \sim N(0,1),
\end{align*}
and that $X \sim \text{unif}([0,1])$. We estimate $r$ using the Nadaraya-Watson estimator with Gaussian kernel of bandwidth $1/500$ and estimate $\sigma(x)$ using again a Nadaraya-Watson estimator as in \cref{appendix:d.nadaraya.watson.variance}. We then construct the score function
\begin{align*}
	S = \frac{Y-r(X)}{\sigma(X)}
\end{align*}
and construct a prediction interval for $S$, then this is used to construct a prediction interval for $Y$. We can see the result in \cref{fig:prediction_interval1D}.

\begin{figure}
	\begin{center}
		\includegraphics[scale=0.5]{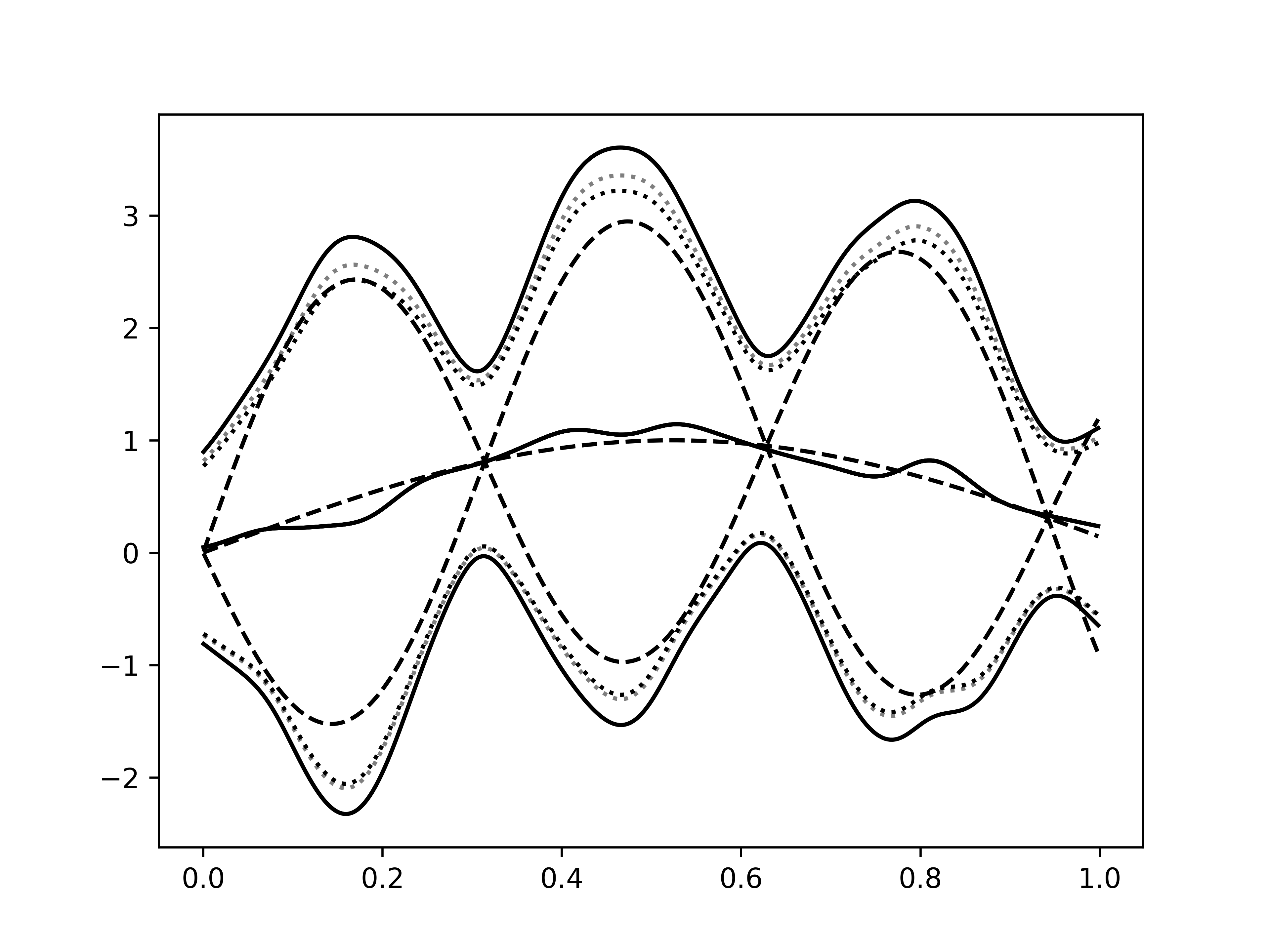}
	\end{center}
	\caption{Simulation of the regression problem. The solid lines are the upper and lower limit of the prediction interval using the $2$-PAC definition and the dotted lines represent the intervals with the $3$-PAC definition for increasing number of samples. The dashed lines represent the true regression function as well as the true conditional quantiles.}
	\label{fig:prediction_interval1D}
\end{figure}

Since our method of constructing prediction intervals is asymmetric we can have a skew distribution for $\epsilon$, in \cref{fig:prediction_interval1D_skew} we see the corresponding result when $\epsilon$ is assumed to be a standardized log-Normal.

We can see an example of the estimated coverage as we keep adding data in \cref{fig:prediction_interval1D_skew_coverage}

\begin{figure}
	\begin{center}
		\includegraphics[scale=0.5]{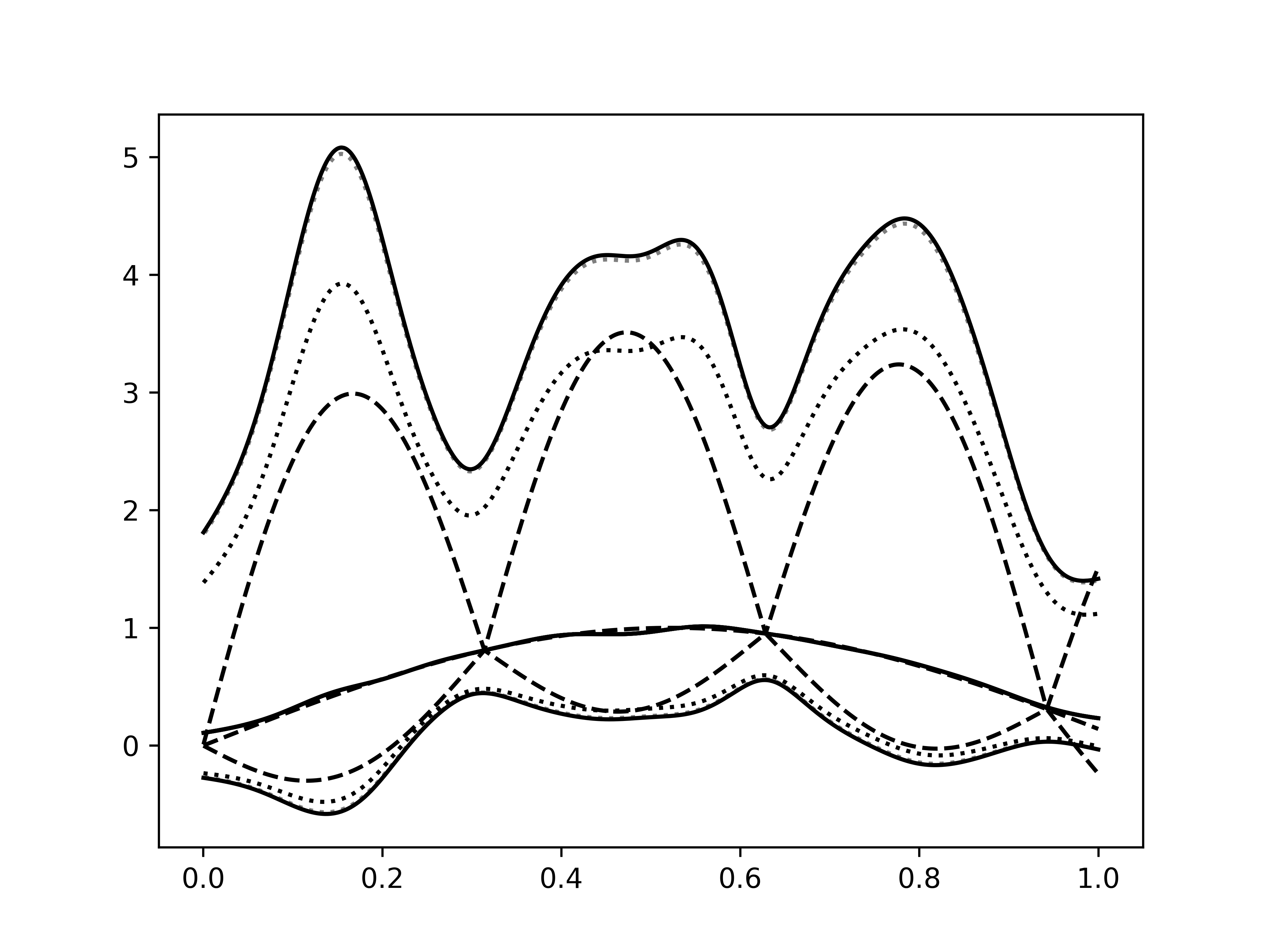}
	\end{center}
	\caption{Simulation of the skew regression problem. The solid lines are the upper and lower limit of the prediction interval using the $2$-PAC definition and the dotted lines represent the intervals with the $3$-PAC definition for increasing number of samples. The dashed lines represent the true regression function as well as the true conditional quantiles.}
	\label{fig:prediction_interval1D_skew}
\end{figure}

\begin{figure}
	\begin{center}
		\includegraphics[scale=0.5]{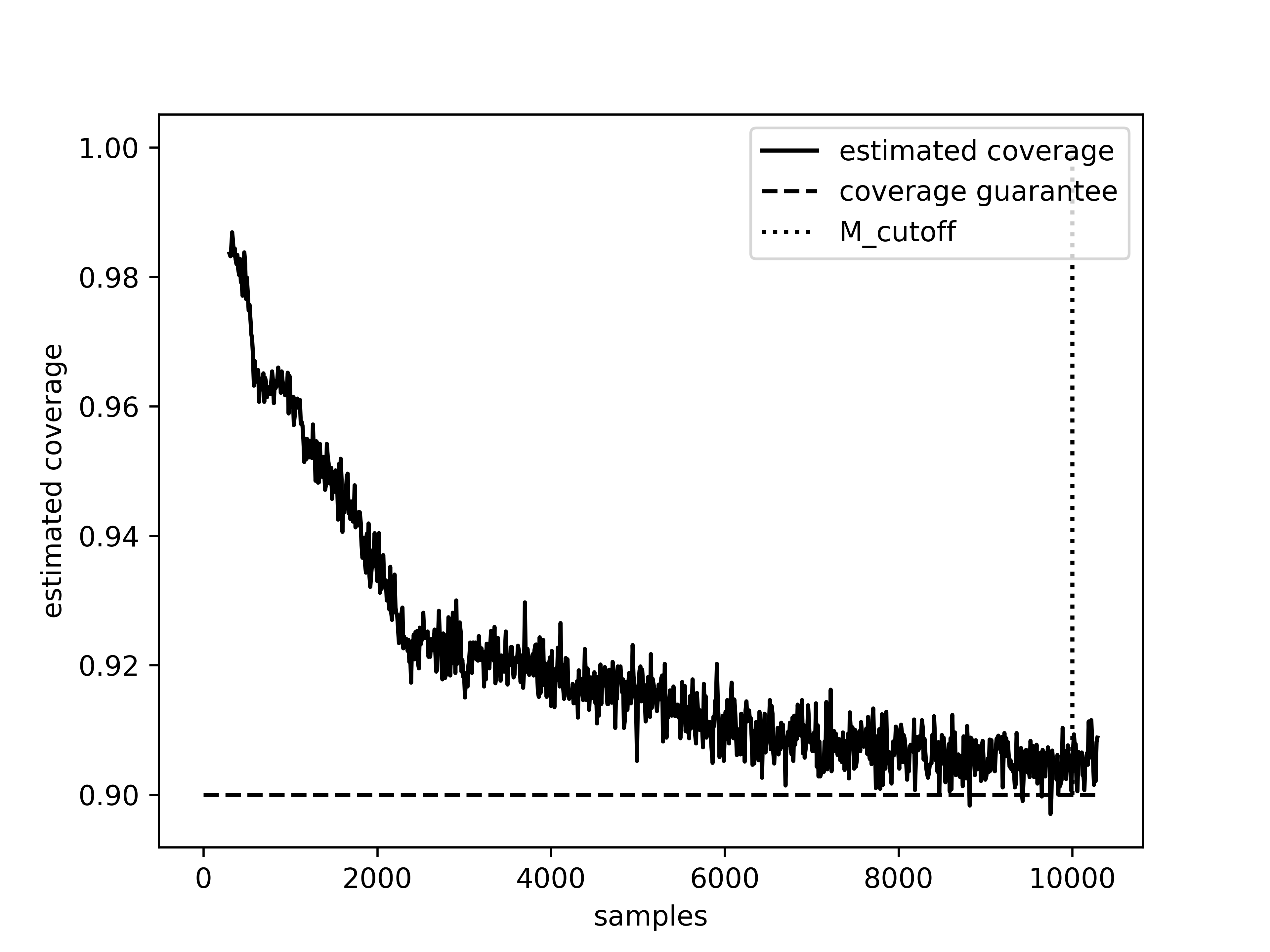}
	\end{center}
	\caption{Estimated coverage of the log-Normal noise using the 3-parameter PAC in \cref{sec:3PAC}.}
	\label{fig:prediction_interval1D_skew_coverage}
\end{figure}

If we instead of relying on the 3-parameter PAC we use \cref{thm:prediction}, we get also the data-dependent upper bound on the true coverage. A resulting simulation can be found in \cref{fig:prediction_interval1D_skew_uniform}.
\begin{figure}
	\begin{center}
		\includegraphics[scale=0.5]{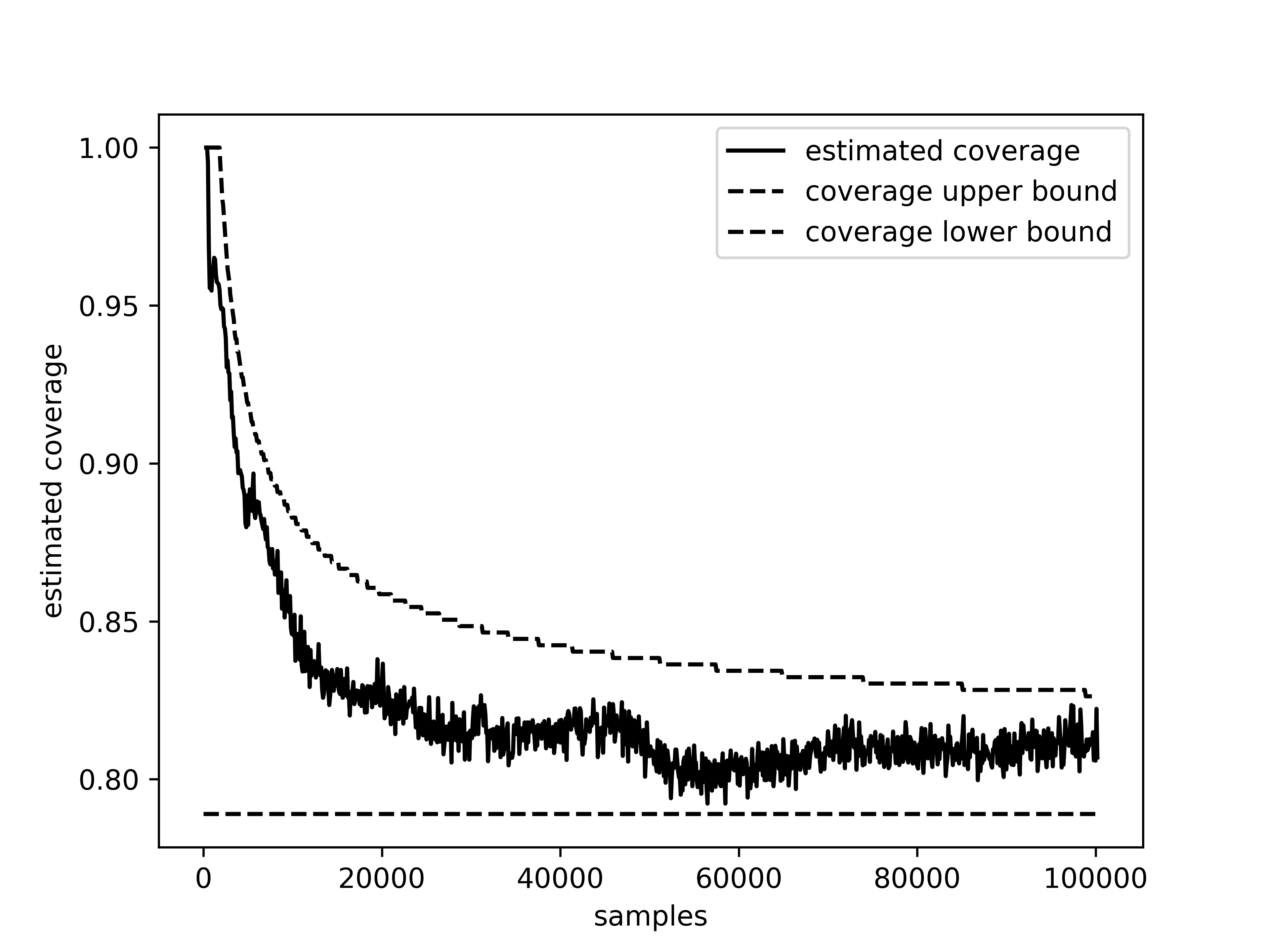}
	\end{center}
	\caption{Estimated coverage using \cref{thm:prediction}.}
	\label{fig:prediction_interval1D_skew_uniform}
\end{figure}

\clearpage

\appendix

\section{Helper lemmas}

\begin{theorem}[Bennett's inequality]\label{thm:bennett}
	Let $X_1,\ldots,X_n$ be i.i.d. random variables with finite variance such that $\P(X_i \leq b) = 1$ with mean zero. Let and $\sigma^2 = \V[X_i]$. Then for any $\epsilon > 0$,
	\begin{align*}
		\P(\overline X_n \geq \epsilon) \leq \exp \left ( - \frac{n \sigma^2}{b^2} h\left ( \frac{b \epsilon }{\sigma^2}\right ) \right )
	\end{align*}
	where $h(u) = (1+u) \log(1+u) - u$ for $u > 0$.
\end{theorem}

\begin{lemma} \label{lem:quantile_confidence}
	Let $Z_1,\ldots$ be an i.i.d. sequence of random variables, let $Q$ and $F$ be the quantile function and distribution function respectively, then
	\begin{align*}
		\P(Q(p) \leq \hat Q_t(p - \epsilon)) \leq \exp \left ( - t p(1-p) h\left ( \frac{\epsilon }{p(1-p)}\right ) \right )
	\end{align*}
	where $h(u) = (1+u) \log(1+u) - u$ for $u > 0$.
	Thus for a given $\delta > 0$ we can find $\epsilon(t,p,\delta)$ such that
	\begin{align*}
		\P(Q(p) \leq \hat Q_t(p - \epsilon(t,p,\delta))) \leq \delta
	\end{align*}
\end{lemma}
\begin{proof}
	First note that by the definition of the quantile function, we have
	\begin{align*}
		\P(Q(p) \leq Q_t(p-\epsilon)) =& \P(\hat F_t(Q(p)) \leq \hat F_t(Q_t(p-\epsilon))) 
		\\
		\leq&
		\P\left (\frac{1}{t} \sum_{i=1}^t \mathbb{I}_{\{Z_i \leq Q(p)\}} \leq p-\epsilon \right)
	\end{align*}
	Now $X_i = \mathbb{I}_{\{Z_i \leq Q(p)\}}$ is an i.i.d. sequence of Bernoulli random variables with probability $p$, thus an application of \cref{thm:bennett} gives the result.
\end{proof}

\section{Nadaraya-Watson for variance} \label{appendix:d.nadaraya.watson.variance}

The basis for the classical Nadaraya-Watson estimator is to first estimate the joint density of the pair $(X,Y) \sim F_{XY}$, where $X,Y \in \R$, using a Kernel density estimator. The Nadaraya-Watson estimator is then nothing but the regression function $\E[Y \mid X]$, where the expectation is taken with respect to the estimated density. Similarly we can estimate the conditional variance, namely we can estimate $\mathbb{V}[Y \mid X]$ by computing it using the estimated density. Here we provide the calculation for reference:

We begin by selecting a kernel $K_h: \R \to \R_+$ such that $\int K_h(x) dx = 1$, $\int xK_h(x)dx = 0$ and $\int x^2 K_h(x) dx = h^2$, i.e. a non-negative function, a classical choice is the Gaussian. Having now an i.i.d. sequence of observations $Z_i = (X_i,Y_i)$, for $i=1,\ldots,n$ we construct the kernel density estimator as
\begin{align*}
	\hat \rho_n(z) = \frac{1}{n}\sum_{i=1}^n K_h(x-X_i)K_h(y-Y_i).
\end{align*}

From this definition we compute
\begin{align*}
	\mathbb{V}_{\hat \rho_n} [Y \mid X] = \int y^2 \frac{\hat \rho_n(z)}{\int \hat \rho(x,y) dy} dy - \E_{\hat \rho_n}[Y|X]^2.
\end{align*}
From which we see that we first need to compute
\begin{align*}
	\int \hat \rho(x,y) dy 
	&= \int \frac{1}{n} \sum_{i=1}^n K_h(x-X_i)K_h(y-Y_i) dy \\
	&= \frac{1}{n} \sum_{i=1}^n K_h(x-X_i)
\end{align*}
Furthermore,
\begin{align*}
	\int y^2 \frac{\hat \rho_n(z)}{\int \hat \rho(x,y) dy} dy 
	&= 
	\frac{1}{n}\sum_{i=1}^n \int y^2 \frac{K_h(x-X_i)K_h(y-Y_i)}{\frac{1}{n} \sum_{j=1}^n K_h(x-X_j)} dy \\
	& = 
	\frac{1}{n}\sum_{i=1}^n \frac{K_h(x-X_i)}{\frac{1}{n} \sum_{j=1}^n K_h(x-X_j)} \int y^2 K_h(y-Y_i)dy
	\\
	&= \frac{\sum_{i=1}^n  K_h(x-X_i) (h^2 + Y_i^2)}{ \sum_{j=1}^n K_h(x-X_j)}
\end{align*}
thus we conclude that the Nadaraya-Watson kernel estimator of variance is
\begin{align*}
	\mathbb{V}_{\hat \rho_n} [Y \mid X] = \frac{\sum_{i=1}^n  K_h(x-X_i) (h^2 + Y_i^2)}{ \sum_{j=1}^n K_h(x-X_j)} - \left ( \frac{\sum_{i=1}^n  K_h(x-X_i) Y_i}{ \sum_{j=1}^n K_h(x-X_j)}\right )^2.
\end{align*}

\printbibliography

\end{document}